\RequirePackage[l2tabu, orthodox]{nag}
\RequirePackage{snapshot}

\documentclass[10pt,onecolumn]{extarticle}

\makeatletter
\if@twocolumn
  \usepackage[dvips,letterpaper,top=0.5in, bottom=0.5in, left=0.75in, right=0.5in,includefoot,heightrounded]{geometry}
\else
  \usepackage[dvips,letterpaper,margin=1in,includefoot,heightrounded]{geometry}
\fi

\usepackage{srcltx}

\usepackage[russian,portuges,english]{babel}

\iflanguage{portuges}
    {\newcommand{\keywordname}{Palavras-chaves}}
    {\newcommand{\keywordname}{Keywords}}

\usepackage{amsmath,amsthm}
\usepackage{amssymb,amsfonts}

\newtheorem{theorem}{Theorem}[section]

\newtheorem{corollary}[theorem]{Corollary}
\newtheorem{remark}[theorem]{Remark}

\newtheorem{algorithm}[theorem]{Algorithm}

\usepackage{abstract}

\usepackage{graphicx}
\usepackage[usenames,dvipsnames,svgnames,x11names]{xcolor}
\usepackage{subfigure}

\usepackage{booktabs}

\usepackage{setspace}
\usepackage{flushend}

\usepackage{cite}

\usepackage{hyperref}
\usepackage[normalem]{ulem}

\usepackage{enumerate}

\usepackage{algorithmic}

\usepackage{listings}

\usepackage[short,12hr]{datetime}
       \usepackage{fouriernc}
\makeatletter

\newcommand{\printtitle}{%
\makeatletter
\if@twocolumn

\twocolumn[%
  \maketitle
  \begin{onecolabstract}
    \myabstract
  \end{onecolabstract}
  \begin{center}
    \small
    \textbf{\keywordname}
    \\\medskip
    \mykeywords
  \end{center}
  \bigskip
]
\saythanks
\else
  \maketitle
  \begin{onecolabstract}
    \myabstract
  \end{onecolabstract}
  \begin{center}
    \small
    \textbf{\keywordname}
    \\\medskip
    \mykeywords
  \end{center}
  \bigskip
  \onehalfspacing
\fi
\makeatother
}

\author{%
Sirani M. Perera%
\thanks{Department of Mathematics,
Embry-Riddle Aeronautical University,
Daytona Beach, FL 32114, USA
(e-mail: pereras2@erau.edu)}
\and
Arjuna Madanayake%
\thanks{Department of Electrical and Computer Engineering, Florida International University, Miami, FL 33174 USA
(e-mail: amadanay@fiu.edu)}
\and
R.~J.~Cintra%
\thanks{Departamento de Estat\'{\i}stica,
Universidade Federal de Pernambuco, Recife, PE
50740540 Brazil
(e-mail: rjdsc@de.ufpe.br)}
}

\title{%
Radix-2 Self-Recursive Sparse Factorizations of  Delay Vandermonde Matrices for Wideband Multi-Beam Antenna Arrays}

\newcommand{\myabstract}{%
This paper presents a self-contained factorization for the Vandermonde matrices associated with  true-time delay based wideband analog  multi-beam beamforming using antenna arrays. The proposed factorization contains sparse and orthogonal matrices. Novel self-recursive radix-2 algorithms for Vandermonde matrices associated with true time delay based  delay-sum filterbanks are presented to reduce the circuit complexity of multi-beam analog beamforming systems. The proposed algorithms for Vandermonde matrices by a vector attain $\mathcal{O}(N \log N)$ delay-amplifier circuit counts.
{Error bounds for the Vandermode matrices associated with true-time delay are established and then analyzed for numerical stability.}
The potential for real-world circuit implementation of the proposed algorithms will be shown through signal flow graphs that are the starting point for high-frequency analog circuit realizations.
}

\newcommand{\mykeywords}{%
Sparse matrices,
Algorithm design and analysis,
Computational complexity,
Accuracy,
Error analysis,
Fast Fourier transforms,
Antenna arrays,
Integrated circuits,
Wireless communication
}

\date{}

\begin{document}

\printtitle

\section{Introduction}
\label{intro}
{T}{he} realization of narrowband discrete Fourier transform (DFT) multi-beams is itself a hard  engineering problem due to circuit complexity of the aperture transceivers. For example,  the phasing network required for forming $N$ beams requires $N^2$ phasing elements. The DFT is a linear operation that maps an $N$-point input signal
$\mathbf{x} = \begin{bmatrix}x[0] & x[1] & \cdots & x[N-1]\end{bmatrix}^\top$
into an $N$-point output signal
$\mathbf{X} = \begin{bmatrix}X[0] & X[1] & \cdots & X[N-1]\end{bmatrix}^\top$
according to the following relationship:
$\mathbf{X} = \mathbf{F}_N \cdot \mathbf{x}$,
where $\mathbf{F}_N$ is the DFT matrix, whose elements are given by $\omega_N^{kl}$, $k,l = 0, 1, \ldots, N-1$, where $\omega_N = \exp\left( - j \frac{2\pi}{N} \right)$ is the $N^{\rm th}$ root of unity and $j=\sqrt{-1}$.

Evaluated by means of direct matrix-vector multiplications, the direct computational complexity of the DFT is in $\mathcal{O}(N^2)$, with $N^2$ complex multiplications and $N(N-1)$ complex additions.
The DFT matrix has been studied for the last 50~years, and there exist a multitude of fast algorithms (collectively called fast Fourier transforms (FFTs)) that compute the DFT using $\mathcal{O}(N\log N)$ operations, which is significantly lower when compared with the direct implementations. The use of a spatial FFT  leads to $N$ independent orthogonal RF beams at $\mathcal{O}(N\log N)$ complexity. In fact, by taking a given FFT algorithm and implementing its ``Twiddle Factors'' (which are intermediate constant complex multiplications found in FFT algorithms) using microwave or analog IC-based phase-shifter implementations has led to the ``Butler Matrix'' type multi-beam array beamformers that are well known in the literature. However, such FFT beams suffer from frequency dependent beam directions. Known as ``beam squint'' because the beam directions are strongly dependent on the temporal frequency of operation, DFT based multi-beam beamformers can only be used for narrowband wireless systems.

The FFT is capable of computing the DFT or its inverse in $\mathcal{O}(N\log N)$ complexity. Therefore, FFT-based multi-beam beamformers are very useful for wireless systems having narrow bandwidth. However, for emerging 5G mmW systems that exploit increasingly wide bandwidths, the beam-squint problem can be significant. For emerging 5G mmW systems that fully exploit the available bandwidth for increasing system capacity, one must utilize the true time-delay based multi-beam beamformers described by its own delay Vandermonde matrix (DVM). The DVM, however, is equal to the DFT only at a single temporal frequency. Therefore, FFT-based factorizations are not applicable for the DVM matrix. In this paper, we describe the complexity of an FFT-like factorization algorithm for the Vandermonde matrices, in order to be able to implement truly wideband multi-beam mmW beamformers based on true-time-delay networks albeit at $\mathcal{O}(N\log  N)$ complexity.

The paper is organized as follows.
{Section} \ref{eetheory} {contains an introduction to complexity metrics of analog and digital parallel computation systems for matrix-vector products.}
Section~\ref{sec:fac} introduces novel self-contained factorizations for Vandermonde matrices and radix-2 algorithms, while in section~\ref{sec:com} we will derive arithmetic complexity and elaborate on numerical results based on the proposed algorithms for Vandermonde matrices.
{Next, section} \ref{sec:err} {analyzes error bounds and stability in computing radix-2 algorithms for Vandermonde matrices having true time-delays.}
In section \ref{sec:sfg} we will present signal flow graphs of the proposed radix-2 algorithm for Vandermonde matrices.
Finally, section~\ref{sec:con} concludes the paper.

\section{Analog Implementations for 5G and Beyond: Quantifying Complexity}
\label{eetheory}

Fast analog radio frequency (RF) integrated circuit (IC) realizations of the beamforming algorithms become necessary when the bandwidths of interest are greater than  a few GHz. For emerging 5G, 6G and beyond, the bandwidths of interest are too high for digital computing solutions to keep up. The solution is to replace digital systems with fast analog implementations of wideband beamforming algorithms, which in turn, requires a revisit to traditional algorithm complexity theory because of differences in analog parallel  architectures compared to conventional digital approaches.  In analog implementations,  the bandwidth effectively sets  the rate at which the analog computation can be updated.  The DVM building block employs true time delays that can be realized using transmission line segments and/or all-pass networks followed by amplification stages.

Let us define DVM fast algorithms as consisting of gain-delay-block (GDB) and addition/subtraction blocks.  Instead of computing the number of multiplications for accessing with arithmetic complexity (as one would do for digital systems), we need to count the number of parallel circuit implementations of GDBs in order to access the circuit complexity of analog parallel algorithms. The larger the number of GDBs, the higher the circuit complexity and hence higher chip area and power consumption.
In analog fast algorithms, the objective is to factorize the original matrix into products of sparse matrices, such that the total number of GDBs is reduced from $\mathcal{O}(N^2)$ to $\mathcal{O}(N\log N).$

{We remark here that the gain is not equivalent to the coefficient multiplication. Although a delay of $t$ is simply multiplication by $e^{-j\omega t}$ in the mathematical sense, it requires a separate true time delay circuit in the analog domain. Hence, the multiplication complexity is different from GBD counts.}

\section{Self-Contained Factorization and Algorithm for Vandermonde Matrices}
\label{sec:fac}
Low complexity and stable algorithms for the delay Vandermonde matrix, $\mathbf{A}_N=[\alpha^{kl}]_{k=1,l=0}^{N,N-1}$, where $\alpha=e^{-j \omega_t \tau}$ and accounts for the phase rotation associated with the delay $\tau$ at frequency $f$, and $\omega_t = 2\pi f$, have been derived through our previous work \cite{SVNA18, VS17, SAR19}. It is important to realize that the matrix elements are integer powers of  $\alpha=e^{-j \omega_t \tau}$ which are functions of the temporal frequency variable $\omega_t$; this is an important distinction from the DFT matrix  where the elements are constants defined as the primitive $N$th  roots of unity. Because integer powers of  $\alpha=e^{-j \omega_t \tau}$ are dependent on $\omega_t$ the DVM frequency responses are functions of two frequency variables: $\omega_x$, which is  typically a spatial variable, and $\omega_t$ which is typically the temporal frequency variable. The DVM matrix frequency responses are defined using the spatial frequency variable $\omega_x$ via 2-D filterbank responses that contain $\omega_t$ as a parameter, and given by the expression for the $k$th filter for $k=0,1,\ldots,N-1$ as $H_k(j\omega_x,j\omega_t)=\sum_i \alpha^{ki}e^{-j\omega_xi}, i=0,1,\ldots,N-1.$ Therefore,  considering both
$\omega_x$ and $\omega_t$ the DVM defines $N$ 2-D frequency responses.

Further, the DVM is the super-class of the DFT matrix without having nice properties like unitary, periodicity, symmetry, and circular shift.
There is no self-contained radix-2 DVM algorithm in the literature. The manuscript \cite{SVNA18} proposes a self-contained sparse factorization of DVM with $\mathcal{O}(N^2)$ arithmetic complexity. The displacement structure of Vandermonde-related matrices is used to derive $\mathcal{O}(N \log^2 N)$ arithmetic complexity algorithms in \cite{GO941, GO942} and an $\mathcal{O}(N)$ arithmetic complexity algorithm in \cite{P17}. The manuscripts \cite{OP00, OA04, Y05} propose $\mathcal{O}(N^2)$ complexity algorithms to compute Vandermonde matrices (having real nodes) by a vector.
The DVM algorithm in \cite{SVNA18} extends the results in \cite{OP00, OA04, Y05} utilizing complex nodes without using displacement equations as in \cite{GO941, GO942, P17}.
Moreover, we have addressed the error bounds and stability of the DVM algorithm in \cite{SVNA18} by filling the gaps in \cite{OP00, OA04, Y05}. The DVM algorithm in \cite{SAR19} is faster than \cite{SVNA18} but does not produce arithmetic complexity of order $\mathcal{O}(N\log  N)$. On the other hand, there are no constraints for nodes of DVM in \cite{SVNA18} as opposed to what we propose here.

In this section, we derive novel self-contained factorization for the Vandermode-type matrices and propose a radix-2 algorithm for the Vandermonde matrices. We will account for the phase rotation associated with delay and frequency in the factorization of Vandermonde matrices.

\subsection{Self-contained Factorization for Vandermonde Matrices}
\label{sub:van}
Algorithms operating on analog signals for computing Vandermonde matrix by a vector can be seen as the evaluation of $(N-1)^{\rm th}$ degree polynomial at $N$ points, albeit using a paralleled analog computing circuit as opposed to a digital realization that must operate on samples and quantized signals. Here we derive self-contained factorization of Vandermonde matrices to obtain efficient continuous-time algorithms for implementation on analog circuits while reducing GDB counts.

One can observe the computation of Vandermonde matrix by a vector with arithmetic complexity $\mathcal{O}(N {\rm} \log^2 N)$ in \cite{GO941, GO942, DHR97}. Here, arithmetic complexity refers to the  number of GDBs in an analog RF-IC circuit implementation, unlike the traditional approach of the number of multipliers and adders in a digital system. There are several mathematical techniques available to derive radix-2 and split-radix FFT algorithms, as described in~\cite{CT65, S86, VL92, JF07, RKH10}.
It has been shown in \cite{P16} that Vandermonde matrices are badly ill-conditioned with a narrow class of exceptions whereas cyclic sequences of nodes are equally spaced on the unit circle $C(0, 1)$.
In here, we propose self-contained and sparse factorization for the well-conditioned Vandermonde matrices and extend the results for $C(0, r)$, where $r > 1$ (i.e. circle of radius $r$ centered at the origin in the complex plane).  The proposed factorizations will then be used to derive fast algorithms while reducing GDB counts.

\begin{theorem}
\label{thm:van}
Let the Vandermonde matrix $\mathbf{V}_{N}=[v_{k}^l]_{k,l=0}^{N-1}$ be defined by equally spaced nodes $\{ v_0, v_1, \ldots, v_{N-1}\}$ on $C(0, 1)$ (in counterclockwise direction) and $N=2^t$ ($t \geq 1$).
Then $\mathbf{ V}_{N}(v_0, v_1, \ldots, v_{N-1})$ can be factored into
\begin{equation}
\begin{aligned}
\mathbf{V}_{N} =P_N^T &
\begin{bmatrix}
 \mathbf{ V}_{\frac{N}{2}}& \\
& \mathbf{ V}_{\frac{N}{2}}
\end{bmatrix}
\begin{bmatrix}
I_{\frac{N}{2}} & \\
& \dot{D}_{\frac{N}{2}}
\end{bmatrix}
\left[
\begin{array}{c|c}
 I_{\frac{N}{2}}&  I_{\frac{N}{2}}\\
\hline\\
 I_{\frac{N}{2}} & -I_{\frac{N}{2}}
\end{array}
\right]
\\ &
\begin{bmatrix}
I_{\frac{N}{2}} & \\
& c \cdot I_{\frac{N}{2}}
\end{bmatrix}
\end{aligned}
\label{eq:VM}
\end{equation}
where $P_N$ is the even-odd permutation matrix, $I_{\frac{N}{2}}$ is the identity matrix, $\dot{D}_{\frac{N}{2}}=\operatorname{diag}[e^{l(\frac{2\pi j}{N})}]_{l=0}^{\frac{N}{2}-1}$, $c=e^{\frac{j\theta N}{2}}$, and $0 \leq \theta < 2 \pi$.
\end{theorem}
\begin{proof}
Let us permute rows of $\mathbf{ V}_{N}$ by multiplying with $P_N$ and write the result as the block matrices:
\begin{equation}
\begin{aligned}
P_N\mathbf{ V}_{N}&=\\
&
\left[
\begin{array}{c|c}
\left[v_{2k}^{l}\right]_{k,l=0}^{\frac{N}{2}-1} & \left[v_{2k}^{\left(\frac{N}{2}+l\right)}\right]_{k,l=0}^{\frac{N}{2}-1} \\
\hline
\left[v_{2k+1}^{l}\right]_{k,l=0}^{\frac{N}{2}-1}& \left[v_{2k+1}^{\left(\frac{N}{2}+l\right)}\right]_{k,l=0}^{\frac{N}{2}-1}
\end{array}
\right]
\end{aligned}
\label{v1eq}
\end{equation}
It is clear that the (1,1) block of the product $P_N\mathbf{ V}_{N}$ is $\mathbf{ V}_{\frac{N}{2}}$. Now, we consider (1,2), (2,1), and (2,2) blocks of $P_N\mathbf{ V}_{N}$ (\ref{v1eq}) and represent each of these by $\mathbf{ V}_{\frac{N}{2}}$ and the product of diagonal matrices.
\\
By (1,2) block of \eqref{v1eq} we get:
\begin{equation}
 \left[v_{2k}^{\left(\frac{N}{2}+l\right)}\right]_{k,l=0}^{\frac{N}{2}-1} =\operatorname{diag}\left[v_{2k}^{\frac{N}{2}}\right]_{k=0}^{\frac{N}{2}-1} \cdot \left[v_{2k}^{l}\right]_{k,l=0}^{\frac{N}{2}-1}
\label{v1eq12}
\end{equation}
Since nodes are equally spaced on $C(0,1)$, we have $v_{2k+1}=v_{2k} \cdot e^{\frac{2\pi j}{N}}$,
for $k=0, 1, \ldots, {\frac{N}{2}-1}$. Now by (2,1) block of \eqref{v1eq} we get:
\begin{equation}
\left[v_{2k+1}^{l}\right]_{k,l=0}^{\frac{N}{2}-1}=\left[v_{2k}^{l}\right]_{k,l=0}^{\frac{N}{2}-1}\cdot \operatorname{diag}[e^{l(\frac{2\pi j}{N})}]_{l=0}^{\frac{N}{2}-1}
\label{v1eq21}
\end{equation}
By (2,2) block of \eqref{v1eq} we get:
\begin{equation}
\begin{aligned}
\left[v_{2k+1}^{\left(\frac{N}{2}+l\right)}\right]_{k,l=0}^{\frac{N}{2}-1} & = -\operatorname{diag}\left[v_{2k}^{\frac{N}{2}}\right]_{k=0}^{\frac{N}{2}-1} \cdot \left[v_{2k}^{l}\right]_{k,l=0}^{\frac{N}{2}-1}\cdot
\\  & \hspace{.2in}
\operatorname{diag}[e^{l(\frac{2\pi j}{N})}]_{l=0}^{\frac{N}{2}-1}
\end{aligned}
\label{v1eq22}
\end{equation}
Thus by \eqref{v1eq12}, \eqref{v1eq21}, and \eqref{v1eq22},
we can state \eqref{v1eq} as:
\begin{equation}
P_N\mathbf{ V}_{N}=\left[
\begin{array}{c|c}
 \mathbf{ V}_{\frac{N}{2}} &  { D}_{\frac{N}{2}} \cdot \mathbf{ V}_{\frac{N}{2}}\\
\hline\\
\mathbf{ V}_{\frac{N}{2}}\cdot \dot{D}_{\frac{N}{2}} & - { D}_{\frac{N}{2}} \cdot \mathbf{ V}_{\frac{N}{2}}\cdot \dot{D}_{\frac{N}{2}}
\end{array}
\right]
\label{v2eq}
\end{equation}
where ${D}_{\frac{N}{2}} =\operatorname{diag}\left[v_{2k}^{\frac{N}{2}}\right]_{k=0}^{\frac{N}{2}-1}$. Let us consider the product of $m$th row of $\mathbf{ V}_{N}$ and $l$th column of $\mathbf{ V}^H_{N}$, where $\mathbf{ V}^H_{N}$ is the conjugate transpose of $\mathbf{ V}_{N}$.
Thus, we have:
\begin{equation}
\begin{aligned}
\mathbf{ V}_{N}(m,:) & \cdot \mathbf{ V}^H_{N}(:,l)
\\
&=
1+v_{m-1}\bar{v}_{l-1}+v_{m-1}^{2} \bar{v}_{l-1}^{2}+ \cdots +v_{m-1}^{(N-1)} \bar{v}_{l-1}^{(N-1)}
\\
& =
\begin{cases}
N, \text{when $m=l$,} \\
0, \text{when $m \neq l$,}
\end{cases}
\end{aligned}
\nonumber
\end{equation}
In the above, the first equality follows as $v_k, \bar{v}_k \in C(0, 1)$ for $k =0, 1, \ldots, N-1$ and the second equality follows as $v_{2k+1}=v_{2k} \cdot e^{\frac{2\pi j}{N}}$. Hence, $\mathbf{ V}_{N}$ is unitary up to scaling by $\frac{1}{\sqrt{N}}$. By using this we can state \eqref{v2eq} as:
\begin{equation}
\begin{aligned}
P_N\mathbf{ V}_{N} &=
\begin{bmatrix}
 \mathbf{ V}_{\frac{N}{2}} & \\
&  \mathbf{ V}_{\frac{N}{2}}
\end{bmatrix}
\\  & \hspace{.2in}
\left[
\begin{array}{c|c}
 \mathbf{ I}_{\frac{N}{2}} &   \frac{2}{N}\cdot\mathbf{ V}^H_{\frac{N}{2}}\cdot { D}_{\frac{N}{2}} \cdot \mathbf{ V}_{\frac{N}{2}}\\
\hline\\
\dot{D}_{\frac{N}{2}} & -  \frac{2}{N}\cdot\mathbf{ V}^H_{\frac{N}{2}} \cdot { D}_{\frac{N}{2}} \cdot \mathbf{ V}_{\frac{N}{2}}\cdot \dot{D}_{\frac{N}{2}}
\end{array}
\right]
\end{aligned}
\label{v3eqn}
\end{equation}
Now let us consider the product $\mathbf{ V}^H_{\frac{N}{2}}\cdot {D}_{\frac{N}{2}} \cdot \mathbf{ V}_{\frac{N}{2}}$ i.e. the product of $m$th row of $\mathbf{ V}^H_{\frac{N}{2}}\cdot { D}_{\frac{N}{2}}$-say $\mathbf{ \hat{V}}_{\frac{N}{2}}$ and $l$th column of $\mathbf{ V}_{\frac{N}{2}}$.
Therefore, we have that
\begin{equation}
\begin{aligned}
& \mathbf{ \hat{V}}_{\frac{N}{2}}(m,:) \cdot  \mathbf{ V}_{\frac{N}{2}}(:,l)
\\
&=
\bar{v}_{0}^{m-1} v_0^{\frac{N}{2}}{v}_{0}^{l-1}+\bar{v}_{2}^{m-1} v_2^{\frac{N}{2}}{v}_{2}^{l-1}+\bar{v}_{4}^{m-1} v_4^{\frac{N}{2}}{z}_{4}^{l-1}
\\ & \hspace{.3in}
+ \cdots +\bar{v}_{N-2}^{m-1} v_{N-2}^{\frac{N}{2}}{v}_{N-2}^{l-1}
\\
& =
\begin{cases}
\displaystyle\sum_{k=0}^{\frac{N}{2}-1}v_{2k}^{\frac{N}{2}},
& \text{when $m=l$,}
\\
0, & \text{when $m \neq l$,}
\end{cases}
\end{aligned}
\nonumber
\end{equation}
In the above, the first equality follows as $v_{2k}, \bar{v}_{2k} \in C(0, 1)$ and the second equality follows as $v_{2k}$ are nodes of $\mathbf{ V}_{\frac{N}{2}}$ and $v_{2k+2}=v_{2k} \cdot e^{\frac{4\pi j}{N}}$. Thus, by following the above one can see the $(m, l)$ entry of $\mathbf{ \hat{V}}_{\frac{N}{2}} \cdot \mathbf{ V}_{\frac{N}{2}}\cdot \dot{D}_{\frac{N}{2}}$ as
\begin{equation}
\begin{aligned}
\text{$(m, l)$ entry of\ }
&\mathbf{ \hat{V}}_{\frac{N}{2}} \cdot \mathbf{ V}_{\frac{N}{2}}\cdot \dot{D}_{\frac{N}{2}}
\\ &
=
\begin{cases}
\left(\displaystyle\sum_{k=0}^{\frac{N}{2}-1}v_{2k}^{\frac{N}{2}}\right)e^{l(\frac{2\pi j}{N})},
\text{when $m=l$,}
\\
0,
\text{when $m \neq l$.}
\end{cases}
\end{aligned}
\nonumber
\end{equation}
Notice that even nodes on $C(0,1)$ can be expressed as $v_{2k} = e^{j\left(\theta+\frac{4\pi k}{N}\right)}$ for $k = 0, 1, \ldots, \frac{N}{2}-1$. Thus, by raising each even node to the power of $\frac{N}{2}$ and taking average we get $c=e^{\frac{j\theta N}{2}}$ where $j^2=-1$. Hence,
\begin{equation}
\mathbf{ V}_{N}=P_N^T
\begin{bmatrix}
 \mathbf{ V}_{\frac{N}{2}}& \\
& \mathbf{ V}_{\frac{N}{2}}
\end{bmatrix}
\left[
\begin{array}{c|c}
 I_{\frac{N}{2}}&  c \cdot I_{\frac{N}{2}}\\
\hline\\
\dot{D}_{\frac{N}{2}} & - c \cdot \dot{D}_{\frac{N}{2}}
\end{array}
\right]
\label{leq}
\end{equation}
and the claim of the theorem follows.
\end{proof}

\begin{remark}
\label{rleq}
The last matrix in the factorization \eqref{leq} has been split into three sparse matrices in \eqref{eq:VM} to reduce the multiplication counts and hence for efficient hardware implementation.
\end{remark}
\begin{corollary}
\label{cthm:van}
Let the Vandermonde matrix $\mathbf{ \tilde{V}}_{N}=[\tilde{v}_{k}^l]_{k,l=0}^{N-1}$ be defined by equally spaced nodes $\{ \tilde{v}_0, \tilde{v}_1, \ldots, \tilde{v}_{N-1}\}$ on $C(0, r)$, where $r >1$ (in counterclockwise direction) and $N=2^t$ ($t \geq 1$). Then $\mathbf{ \tilde{V}}_{N}(\tilde{v}_0, \tilde{v}_1, \ldots, \tilde{v}_{N-1})$ can be factored into
\begin{equation}
\mathbf{ \tilde{V}}_{N}=\mathbf{ V}_N \tilde{\mathbf{D}}_N
\label{ceq:VM}
\end{equation}
where $\tilde{\mathbf{D}}_N=\operatorname{diag} [r^{l}]_{l=0}^{N-1}$ and $\mathbf{ V}_N$ is defined via \eqref{eq:VM}.
\end{corollary}

\begin{proof}
This is trivial as $\tilde{v}_k=r \cdot v_k$  for $k=0, 1, \ldots, N-1$.
\end{proof}

The following self-contained factorization for the Vandermonde matrices is proposed in connection to the phase rotation associated with delay $\tau$ and frequency $\omega_t=2\pi f$.

\begin{theorem}
\label{thm:van2}
Let the Vandermonde matrix $\mathbf{ V}_{N}=[v_{k}^l]_{k,l=0}^{N-1}$ be defined by equally spaced nodes $\{ v_0, v_1, \ldots, v_{N-1}\}$ on $C(0, 1)$ (in clockwise direction) and $N=2^t$ ($t \geq 1$). Then $\mathbf{ V}_{N}(v_0, v_1, \ldots, v_{N-1})$ can be factored into
\begin{equation}
\begin{aligned}
\mathbf{ V}_{N}=P_N^T &
\begin{bmatrix}
 \mathbf{ V}_{\frac{N}{2}}& \\
& \mathbf{ V}_{\frac{N}{2}}
\end{bmatrix}
\begin{bmatrix}
I_{\frac{N}{2}} & \\
& \bar{\dot{D}}_{\frac{N}{2}}
\end{bmatrix}
\left[
\begin{array}{c|c}
 I_{\frac{N}{2}}&  I_{\frac{N}{2}}\\
\hline\\
 I_{\frac{N}{2}} & -I_{\frac{N}{2}}
\end{array}
\right]
\\ &
\begin{bmatrix}
I_{\frac{N}{2}} & \\
& \bar{c} \cdot I_{\frac{N}{2}}
\end{bmatrix}
\end{aligned}
\label{eq:cVM}
\end{equation}
where $I_{\frac{N}{2}}$ is the identity matrix, $\bar{\dot{D}}_{\frac{N}{2}}=\operatorname{diag}[e^{-l(\frac{2\pi j}{N})}]_{l=0}^{\frac{N}{2}-1}$, $\bar{c}=e^{-\frac{j\theta N}{2}}$, and  $\theta=2 \pi f \tau = \omega_t\tau$,  s.t. $0 \leq \theta < 2 \pi$ .
\end{theorem}
\begin{proof}
The proof follows similar lines as that of Theorem \ref{thm:van}, except  $\bar{\dot{D}}_{\frac{N}{2}}=\operatorname{diag}[e^{-l(\frac{2\pi j}{N})}]_{l=0}^{\frac{N}{2}-1}$ instead of $\dot{D}_{\frac{N}{2}}=\operatorname{diag}[e^{l(\frac{2\pi j}{N})}]_{l=0}^{\frac{N}{2}-1}$ and $\bar{c}$ instead of $c$.
\end{proof}
\begin{remark}
Theorem \ref{thm:van2} has proposed a self-contained factorization, as opposed to a scaled DFT matrix. If one chooses to scale DFT matrices to factor $\mathbf{ V}_{N}$, it results in the computation of small complex numbers and leads to zero matrices \cite{H96}. The proposed factorization for $\mathbf{ V}_{N}$ in \eqref{eq:cVM} overcomes this barrier.
\end{remark}

\begin{corollary}
\label{cthm:van2}
Let the Vandermonde matrix $\mathbf{ \tilde{V}}_{N}=[\tilde{v}_{k}^l]_{k,l=0}^{N-1}$ be defined by equally spaced nodes $\{ \tilde{v}_0, \tilde{v}_1, \ldots, \tilde{v}_{N-1}\}$ on $C(0, r)$, where $r >1$ (in clockwise direction) and $N=2^t$ ($t \geq 1$). Then $\mathbf{ \tilde{V}}_{N}(\tilde{v}_0, \tilde{v}_1, \ldots, \tilde{v}_{N-1})$ can be factored into
\begin{equation}
\mathbf{ \tilde{V}}_{N}=\mathbf{ V}_N \tilde{\mathbf{D}}_N
\label{ceq:VM2}
\end{equation}
where $\tilde{\mathbf{D}}_N=\operatorname{diag} [r^{l}]_{l=0}^{N-1}$ and $\mathbf{ V}_N$ is defined via \eqref{eq:cVM}.
\end{corollary}

\begin{proof}
This is trivial as $\tilde{v}_k=r \cdot v_k$ for $k=0, 1, \ldots, N-1$.
\end{proof}

\begin{remark}
\label{cDFT}
When $\theta=0$ and $r=1$, the proposed factorization for the Vandermode matrices given in Theorem \ref{thm:van2}, reduces to the well known self-contained DFT matrix factorization \cite{CT65, Y68, VL92, S94}. Thus, we can use this property to define a delay Vandermonde matrix to solve the beam squint problem as well as allow high-speed analog realizations for future high bandwidth applications where the slowing down of Moore's law prevents the adoption of digital parallel processing architectures.
\end{remark}

\subsection{Self-recursive Algorithms for Vandermonde Matrices}
\label{sub:algovan}
In the following, we will state self-recursive radix-2 algorithms for Vandermonde matrices with the help of the Theorem \ref{thm:van}, Theorem \ref{thm:van2}, Corollary \ref{cthm:van} and Corollary \ref{cthm:van2}. Let us call the corresponding algorithms $\bf{vanc(N)}$, $\bf{vancc(N)}$, $\bf{vancr(N)}$, and $\bf{vanccr(N)}$ respectively, e.g., the acronym $\bf{vancr(N)}$ was selected to refer to the factorization for the Vandemode matrices having clockwise nodes on the circle of radius $r$. We use the following notation for the inputs of the algorithms i.e. $N$ for the size of the matrices, $\theta$, where $0 \leq \theta < 2\pi$, for the angle of rotation from the positive real axis (positive or negative based on counterclockwise or clockwise direction), $r$ for the magnitude, and $\mathbf{z}$ for the input vector.

Before stating algorithms, let us use the following notation to denote sparse matrices which will be used hereafter for $N \geq 4$.

\begin{equation}
\begin{matrix}
\hat{\mathbf{D}}_N=\begin{bmatrix}
I_{\frac{N}{2}} & \\
& \dot{D}_{\frac{N}{2}}
\end{bmatrix},\:\:\:  \check{\mathbf{D}}_N=\begin{bmatrix}
I_{\frac{N}{2}} & \\
& \bar{\dot{D}}_{\frac{N}{2}}
\end{bmatrix}
\\\\
\hat{\mathbf{I}}_N= \left[
\begin{array}{c|c}
 I_{\frac{N}{2}}&  I_{\frac{N}{2}}\\
\hline\\
 I_{\frac{N}{2}} & -I_{\frac{N}{2}}
\end{array}
\right],\\\\
\mathbf{C}_N=\begin{bmatrix}
I_{\frac{N}{2}} & \\
& c \cdot I_{\frac{N}{2}}
\end{bmatrix}, {\rm and}\:\:\: \bar{\mathbf{C}}_N=\begin{bmatrix}
I_{\frac{N}{2}} & \\
& \bar{c} \cdot I_{\frac{N}{2}}
\end{bmatrix}
\end{matrix}
\label{feqa}
\end{equation}

\begin{algorithm} $\mathbf{vancc(z, N)}$\\
Input: $N = 2^t$ ($t \geq 1$), $N_1=\frac{N}{2}$, $\theta$, and $\mathbf{ z} \in \mathbb{R}^n {\:\:\rm or\:\:} \mathbb{C}^n$.
\begin{enumerate}
\item If $N=2$, then \\
\hspace{.1in} $\mathbf{ y}=\begin{bmatrix}
1 & e^{j\theta}\\
1 & -e^{j\theta}
\end{bmatrix} \mathbf{ z}.$
\item If $N \geq 4$, then \\
 \hspace{.1in} $\mathbf{ u}:=\mathbf{ C}_N \mathbf{ z}$,\\
 \hspace{.1in} $\mathbf{ v}:= \hat{\mathbf{ I}}_N \mathbf{ u}$,\\
 \hspace{.1in} $\mathbf{ w}:= \hat{\mathbf{ D}}_N \mathbf{ v}$,\\
 \hspace{.1in} $\mathbf{ s1}:=\mathbf{ vancc} \left(\left[w_i \right]_{i=0}^{N_1-1}, N_1 \right)$,\\
 \hspace{.1in} $\mathbf{ s2}:=\mathbf{ vancc} \left(\left[w_i \right]_{i=N_1}^{N}, N_1 \right)$,\\
 \hspace{.1in} $\mathbf{ y}:=\mathbf{ P}_N^T \left(\mathbf{ s1}^T, \mathbf{ s2}^T \right)^T$.
\end{enumerate}
Output: $\mathbf{ y}=\mathbf{ V}_{N}\mathbf{ z}$.
\end{algorithm}

\begin{algorithm} $\mathbf{vanc(z, N)}$\\
Input: $N = 2^t$ ($t \geq 1$), $N_1=\frac{N}{2}$, $\theta$, and $\mathbf{ z} \in \mathbb{R}^n {\:\:\rm or\:\:} \mathbb{C}^n$.
\begin{enumerate}
\item If $N=2$, then \\
\hspace{.1in} $\mathbf{ y}=\begin{bmatrix}
1 & e^{-j\theta}\\
1 & -e^{-j\theta}
\end{bmatrix} \mathbf{ z}.$
\item If $N \geq 4$, then \\
\hspace{.1in} $\mathbf{ u}:=\mathbf{ \bar{C}}_N \mathbf{ z}$,\\
 \hspace{.1in} $\mathbf{ v}:= \hat{\mathbf{ I}}_N \mathbf{ u}$,\\
 \hspace{.1in} $\mathbf{ w}:= \check{\mathbf{ D}}_N \mathbf{ v}$,\\
  \hspace{.1in} $\mathbf{ s1}:=\mathbf{ vanc} \left(\left[w_i \right]_{i=0}^{N_1-1}, N_1 \right)$,\\
 \hspace{.1in} $\mathbf{ s2}:=\mathbf{ vanc} \left(\left[w_i \right]_{i=N_1}^{N}, N_1 \right)$,\\
 \hspace{.1in} $\mathbf{ y}:=\mathbf{ P}_N^T \left(\mathbf{ s1}^T, \mathbf{ s2}^T \right)^T$.
\end{enumerate}
Output: $\mathbf{ y}=\mathbf{ V}_{N}\mathbf{ z}$.
\end{algorithm}

\begin{algorithm} $\mathbf{vanccr(z, N)}$\\
Input: $N = 2^t$ ($t \geq 1$), $N_1=\frac{N}{2}$,  $r$, $\theta$, and $\mathbf{ z} \in \mathbb{R}^n {\:\:\rm or\:\:} \mathbb{C}^n$.
\begin{enumerate}
\item If $N=2$, then \\
\hspace{.1in} $\mathbf{ y}=\begin{bmatrix}
1 & re^{j\theta}\\
1 & -re^{j\theta}
\end{bmatrix} \mathbf{ z}.$
\item If $N \geq 4$, then \\
 \hspace{.1in} $\mathbf{ u}:=\mathbf{ \tilde{D}}_N \mathbf{ z}$,\\
 \hspace{.1in} $\mathbf{ y}:= \mathbf{ vancc} \left(\left[u_i \right]_{i=0}^{N-1}, N \right)$.
\end{enumerate}
Output: $\mathbf{ y}=\mathbf{ \tilde{V}}_{N}\mathbf{ z}$.
\end{algorithm}

\begin{algorithm} $\mathbf{vancr(z, N)}$\\
Input: $N = 2^t$ ($t \geq 1$), $N_1=\frac{N}{2}$, $r$, $\theta$, and $\mathbf{ z} \in \mathbb{R}^n {\:\:\rm or\:\:} \mathbb{C}^n$.
\begin{enumerate}
\item If $N=2$, then \\
\hspace{.1in} $\mathbf{ y}=\left[\begin{array}{ll}
1 & re^{-j\theta}\\
1 & -re^{-j\theta}
\end{array}\right] \mathbf{ z}.$
\item If $N \geq 4$, then \\
\hspace{.1in} $\mathbf{ u}:=\mathbf{ \tilde{D}}_N \mathbf{ z}$,\\
 \hspace{.1in} $\mathbf{ y}:=\mathbf{ vanc} \left(\left[u_i \right]_{i=0}^{N-1}, N \right)$.
\end{enumerate}
Output: $\mathbf{ y}=\mathbf{ \tilde{V}}_{N}\mathbf{ z}$.
\end{algorithm}

\section{Analog GDB-Complexity}
\label{sec:com}
The number of additions and multiplications required to carry out a computation is called the arithmetic complexity in a digital computing system. Here, because our intention is to realize these algorithms as high-speed analog computing circuits operating at RF, we use the modified arithmetic complexity metric where we are counting the number of GDBs instead of multipliers. In this section, the GDB counts of the proposed self-contained factorization for the Vandermonde matrices via algorithms $\mathbf{vanc(z, N)}$, $\mathbf{vancc(z, N)}$, $\mathbf{vancr(z, N)}$, and $\mathbf{vanccr(z, N)}$ will be addressed. The direct  analog computation {of the}
Vandermonde matrix by a vector $\mathbf{ z} \in \mathbb{C}$ in the usual way requires $\mathcal{O}(N^2)$ GDB circuits to be realized in parallel in the RF-IC analog computing device.

{However, we will show in this section that the proposed
 self-recursive radix-2 algorithms can be utilized to compute Vandermonde matrices by a vector with $\mathcal{O}(N\: {\log}\: N)$ GDB counts.}

This is a dramatic circuit complexity reduction of Vandermonde matrices by a vector in the literature. Although the computation speed is still the same, the new factorization reduces chip area and power consumption due to the smaller amount of GDB circuits that have to be physically realized on the analog computing device.

\subsection{GDB Counts of Analog Fast Algorithms for Vandermonde Matrices}
\label{sub:comvan}
Here we analyze the analog GDB counts of the radix-2 algorithms for Vandermonde matrices presented in Section \ref{sub:van}. Let us denote the number of complex/real additions (say $\#a\mathbb{C}$/$\#a\mathbb{R}$ respectively) and complex/real multiplications (say $\#m\mathbb{C}/\#m\mathbb{R}$ respectively) required to compute $\mathbf{ y}=\mathbf{ V}_{N}\mathbf{ z}$ and $\mathbf{ y}=\mathbf{ \tilde{V}}_{N}\mathbf{ z}$ having $\mathbf{ z} \in \mathbb{C}^N$ or $\mathbb{R}^N$. We do not count multiplication by $\pm 1$ and permutation.

Let us first analyze the complex GDB counts of the radix-2 algorithms for Vandermonde matrices by a complex input vector.  We recall that the GDBs implement a complex multiplication defined in the frequency domain $\omega_t$ which requires a time-domain delay to implement on the DVM signal flow graphs. We recall that the independent frequency variable of the DVM is $\omega_x$ and that $\omega_t$ is the temporal frequency parameter associated with the matrix elements $\alpha.$ This is why the complex multiplication operations, which contain $e^{-j\omega_t\tau}$ terms, must in practice be realized in the time domain using time-delays.

\begin{theorem}
\label{thm:cvan}
Let $N=2^t (\geq 2)$ and $\theta$ be given. The complex GDB counts of the proposed $\mathbf{vancc(z, N)}$ algorithm with $\mathbf{ z} \in \mathbb{C}^N$ is given by
\begin{align}
\#a\mathbb{C}(VanCC, N) &=Nt,
\nonumber \\
\#m\mathbb{C}(VanCC, N) &= Nt-N+1.
\label{amvan}
\end{align}
\end{theorem}
\begin{proof}
Referring to the algorithm $\mathbf{vancc(z, N)}$, we get
\begin{equation}
\begin{aligned}
\#a\mathbb{C}(\textrm{VanCC}, N)  &= 2\cdot \#a\mathbb{C}\left(\textrm{VanCC}, \frac{N}{2} \right) +  \#a\mathbb{C}\left(\hat{\mathbf{ D}}_N \right)
\\ & \hspace{.1in} +\#a\mathbb{C}\left(\hat{\mathbf{ I}}_N \right)+\#a\mathbb{C}\left(\mathbf{ C}_N \right)
\end{aligned}
\label{vanaim}
\end{equation}
By following the structures of $\hat{\mathbf{D}}_N$, $\hat{\mathbf{I}}_N$ and $\mathbf{C}_N$,
\begin{equation}
\begin{split}
\#a\mathbb{C}\left(\hat{\mathbf{ D}}_N \right)= 0,
\quad
&\#m\mathbb{C}\left(\hat{\mathbf{ D}}_N \right) = \frac{N}{2}-1
\\
\#a\mathbb{C}\left(\hat{\mathbf{ I}}_N \right)= N,
\quad
&\#m\mathbb{C}\left(\hat{\mathbf{ I}}_N \right) = 0
\\
\#a\mathbb{C}\left(\mathbf{ C}_N \right)= 0,
\quad
&\#m\mathbb{C}\left(\mathbf{ C}_N \right) = \frac{N}{2}
\end{split}
\label{Ddot}
\end{equation}
Thus by using the above, we could state \eqref{vanaim} as the first order difference equation with respect to $t \geq 2$
 \[
\#a\mathbb{C}(\textrm{VanCC}, 2^t)  - 2\cdot \#a\mathbb{C}\left(\textrm{VanCC}, 2^{t-1}\right) = 2^{t}.
\]
Solving the above difference equation using the initial condition $\#a\mathbb{C}(\textrm{VanCC}, 2)=2$, we can obtain
\[
\#a\mathbb{C}(\textrm{VanCC}, 2^t) =Nt.
\]
Now by using the algorithm $\mathbf{vancc(z, N)}$ and \eqref{Ddot}, we could obtain another first order difference equation with respect to $t \geq 2$
 \[
\#m\mathbb{C}(\textrm{VanCC}, 2^t)  - 2\cdot \#m\mathbb{C}\left(\textrm{VanCC}, 2^{t-1}\right) =  2^{t}-1.
\]
Solving the above difference equation using the initial condition $\#m\mathbb{C}(\textrm{VanCC}, 2)=1$, we can obtain
\[
\#m\mathbb{C}(\textrm{VanCC}, 2^t) =Nt-N+1.
\]
\end{proof}

\begin{corollary}
\label{cthm:cvan}
Let $N=2^t (\geq 2)$, $r$ and $\theta$ be given. The complex GDB counts of the proposed $\mathbf{vanccr(z, N)}$ algorithm with $\mathbf{ z} \in \mathbb{C}^N$ is given by
\begin{align}
\#a\mathbb{C}(VanCCR, N) &=Nt,
\nonumber \\
\#m\mathbb{C}(VanCCR, N) &= Nt-\frac{1}{2}N.
\label{camvan}
\end{align}
\end{corollary}
\begin{proof}
The multiplication of the diagonal matrix $\tilde{\mathbf{D}}_N$ with a complex input counts no addition and $\frac{N}{2}-1$ multiplications. Thus by using $\mathbf{vanccr(z, N)}$ algorithm and GDB counts in \eqref{amvan}, the complex GDB counts can be obtained as in \eqref{amvan}.
\end{proof}
\begin{theorem}
\label{thm:ccvan}
Let $N=2^t (\geq 2)$ and $\theta$ be given. The complex GDB counts of the proposed $\mathbf{vanc(z, N)}$ algorithm with $\mathbf{ z} \in \mathbb{C}^N$ is given by
\begin{align}
\#a\mathbb{C}(VanC, N) &=Nt,
\nonumber \\
\#m\mathbb{C}(VanC, N) &= Nt-N+1.
\label{amvanc}
\end{align}
\end{theorem}
\begin{proof}
The proof follows similar lines as that of Theorem \ref{thm:cvan} except $\check{\mathbf{D}}_N$ instead of $\hat{\mathbf{D}}_N$ and $\bar{\mathbf{C}}_N$ instead of $\mathbf{C}_N$.
\end{proof}

\begin{corollary}
\label{cthm:ccvan}
Let $N=2^t (\geq 2)$, $r$ and $\theta$ be given. The complex GDB counts of the proposed $\mathbf{vancr(z, N)}$ algorithm with $\mathbf{ z} \in \mathbb{C}^N$ is given by
\begin{align}
\#a\mathbb{C}(VanCR, N) &=Nt,
\nonumber \\
\#m\mathbb{C}(VanCR, N) &= Nt-\frac{1}{2}N.
\label{camvanc}
\end{align}
\end{corollary}
\begin{proof}
The multiplication of the diagonal matrix $\tilde{\mathbf{D}}_N$ with a complex input counts no addition and $\frac{N}{2}-1$ multiplications. Thus by using $\mathbf{vancr(z, N)}$ algorithm and GDB counts in \eqref{amvanc}, the complex GDB counts can be obtained as in \eqref{camvanc}.
\end{proof}

Let us analyze the real GDB counts of the radix-2 algorithms for Vandermonde matrices by a real input vector. Here we count the multiplication of two complex numbers with 2 real additions and 4 real multiplications.

\begin{theorem}
\label{Rthm:cvan}
Let $N=2^t (\geq 2)$ and $\theta$ be given. The real GDB counts of the proposed $\mathbf{vancc(z, N)}$ algorithm with $\mathbf{ z} \in \mathbb{R}^N$ is given by
\begin{align}
\#a\mathbb{R}(VanCC, N) &=Nt,
\nonumber \\
\#m\mathbb{R}(VanCC, N) &= 2Nt-\frac{5}{2}N+2.
\label{Ramvan}
\end{align}
\end{theorem}
\begin{proof}
Referring to the algorithm $\mathbf{vancc(z, N)}$, we get
\begin{equation}
\begin{aligned}
\#m\mathbb{R}(\textrm{VanCC}, N) & = 2\cdot \#m\mathbb{R}\left(\textrm{VanCC}, \frac{N}{2} \right) +  \#m\mathbb{R}\left(\hat{\mathbf{ D}}_N \right)
\\ & \hspace{.1in}+\#m\mathbb{R}\left(\hat{\mathbf{ I}}_N \right)+\#m\mathbb{R}\left(\mathbf{ C}_N \right)
\end{aligned}
\label{Rvanaim}
\end{equation}
By following the structures of $\hat{\mathbf{D}}_N$, $\hat{\mathbf{I}}_N$ and $\mathbf{C}_N$,
\begin{equation}
\begin{split}
\#a\mathbb{R}\left(\hat{\mathbf{ D}}_N \right)= 0,
\quad
&\#m\mathbb{R}\left(\hat{\mathbf{ D}}_N \right) = N-2,
\\
\#a\mathbb{R}\left(\hat{\mathbf{ I}}_N \right)= N,
\quad
&\#m\mathbb{R}\left(\hat{\mathbf{ I}}_N \right) = 0,
\\
\#a\mathbb{R}\left(\mathbf{ C}_N \right)= 0,
\quad
&\#m\mathbb{R}\left(\mathbf{ C}_N \right) = N.
\end{split}
\label{RDdot}
\end{equation}
\end{proof}
Thus by using the above, we could state \eqref{Rvanaim} as the first order difference equation with respect to $t \geq 2$
 \[
\#m\mathbb{R}(\textrm{VanCC}, 2^t)  - 2\cdot \#m\mathbb{R}\left(\textrm{VanCC}, 2^{t-1}\right) = 2 \cdot 2^{t}-2.
\]
Solving the above difference equation using the initial condition $\#m\mathbb{R}(\textrm{VanCC}, 2)=1$, we can obtain
\[
\#m\mathbb{R}(\textrm{VanCC}, 2^t) =2Nt-\frac{5}{2}N+2
\]
Now by using the algorithm $\mathbf{vancc(z, N)}$ and \eqref{Ddot}, we could obtain another first order difference equation with respect to $t \geq 2$
 \[
\#a\mathbb{R}(\textrm{VanCC}, 2^t)  - 2\cdot \#a\mathbb{R}\left(\textrm{VanCC}, 2^{t-1}\right) = 2^{t}.
\]
Solving the above difference equation using the initial condition $\#a\mathbb{R}(\textrm{VanCC}, 2)=2$, we can obtain
\[
\#a\mathbb{R}(\textrm{VanCC}, 2^t) =Nt.
\]

\begin{corollary}
\label{cRthm:cvan}
Let $N=2^t (\geq 2)$, $r$ and $\theta$ be given. The real GDB counts of the proposed $\mathbf{vanccr(z, N)}$ algorithm with $\mathbf{ z} \in \mathbb{R}^N$ is given by
\begin{align}
\#a\mathbb{R}(VanCCR, N) &=Nt,
\nonumber \\
\#m\mathbb{R}(VanCCR, N) &= 2Nt-\frac{3}{2}N+1.
\label{cRamvan}
\end{align}
\end{corollary}
\begin{proof}
$\tilde{\mathbf{D}}_N$ is a diagonal matrix with real entries so the number of additions will remain the same as in \eqref{Ramvan} while the number of multiplications will be increased by $N-1$ in \eqref{Ramvan}.
\end{proof}

\begin{theorem}
\label{Rthm:ccvan}
Let $N=2^t (\geq 2)$ and $\theta$ be given. The real GDB counts of the proposed $\mathbf{vanc(z, N)}$ algorithm with $\mathbf{ z} \in \mathbb{R}^N$ is given by
\begin{align}
\#a\mathbb{R}(VanC, N) &=Nt,
\nonumber \\
\#m\mathbb{R}(VanC, N) &= 2Nt-\frac{5}{2}N+2.
\label{Ramvanc}
\end{align}
\end{theorem}
\begin{proof}
The proof follows similar lines as that of Theorem \ref{Rthm:cvan} except $\check{D}_N$ instead of $\hat{\mathbf{D}}_N$ and $\bar{\mathbf{C}}_N$ instead of $\mathbf{C}_N$.
\end{proof}

\begin{corollary}
\label{cRthm:ccvan}
Let $N=2^t (\geq 2)$, $r$ and $\theta$ be given. The real GDB counts of the proposed $\mathbf{vancr(z, N)}$ algorithm with $\mathbf{ z} \in \mathbb{R}^N$ is given by
\begin{align}
\#a\mathbb{R}(VanCR, N) &=Nt,
\nonumber \\
\#m\mathbb{R}(VanCR, N) &= 2Nt-\frac{3}{2}N+1.
\label{cRamvanc}
\end{align}
\end{corollary}
\begin{proof}
$\tilde{\mathbf{D}}_N$ is a diagonal matrix with real entries so the number of additions will remain the same as in \eqref{Ramvanc} while the number of multiplications will be increased by $N-1$ in \eqref{Ramvanc}.
\end{proof}

\subsection{Numerical Results}
\label{sec:numcom}
Here we provide numerical results for the GDB counts of the proposed radix-2 algorithms $\mathbf{ vanc(z, N)}$, $\mathbf{ vancc(z, N)}$, $\mathbf{ vancr(z, N)}$, and $\mathbf{ vanccr(z, N)}$. We consider the direct computation of Vandermonde matrices $\mathbf{ V}$ and $\mathbf{ \tilde{V}}$ by the vector $\mathbf{ z} \in \mathbb{C}^N$ with $N(N-1)$ complex additions and multiplications (note that $\mathbf{ V}$ and $\mathbf{ \tilde{V}}$ have 1's along the first column so we counted the multiplication count as $N(N-1)$ as opposed to $N^2$). Also, the direct computation of Vandermonde matrices $\mathbf{ V}$ and $\mathbf{ \tilde{V}}$ by the vector $\mathbf{ z} \in \mathbb{R}^N$ is taken as $N(2N-1)$ real additions and $2N(N-1)$ real multiplications (since $v_k = e^{-j\left(\theta+\frac{2\pi k}{N}\right)}$ we have considered on computing the powers of nodes using $v_k^l = e^{-jl\left(\theta+\frac{2\pi k}{N}\right)}$ for $l=2, 3, \cdots, N-1$). Note that we have not counted the multiplication by 1 in the Vandermonde matrices.
The numerical results for the GDB counts of the proposed algorithms $\mathbf{ vanc(z, N)}$, $\mathbf{ vancr(z, N)}$, $\mathbf{ vancc(z, N)}$,and $\mathbf{ vanccr(z, N)}$ with corresponding matrices $\mathbf{ V}_N$ and $\mathbf{ \tilde{V}}_N$  varying sizes from $4 \times 4$ to $4096 \times 4096$ are shown in Tables \ref{tbl:Vancost}, \ref{tbl:RVancost}, and \ref{tbl:RRVancost}.

\begin{table*}
\centering
\caption{Complex GDB counts of the proposed radix-2 algorithms (i.e. $\mathbf{ vanc(z, N)}$, $\mathbf{vancr(z, N)}$, $\mathbf{vancc(z, N)}$,and $\mathbf{vanccr(z, N)}$) vs Direct computation}
\begin{tabular}{ | l | l | l | l | l |}
    \hline
   $N$ & Direct  & $\#a\mathbb{C}(VanC, N)/$ & $\#m\mathbb{C}(VanC, N)/$ & $\#m\mathbb{C}(VanCR, N)$/ \\
 & Add/Multi  & $\#a\mathbb{C}(VanCR, N)$/ & $\#m\mathbb{C}(VanCC, N)$ &  $\#m\mathbb{C}(VanCCR, N)$\\
 &  & $\#a\mathbb{C}(VanCC, N)/$ &  &  \\
 &  & $\#a\mathbb{C}(VanCCR, N)$ &  &  \\
\hline \hline
4&  12  & 8  & 5  &    6\\ \hline
8& 56 &  24 &   17&  20 \\ \hline
16  &240  & 64  & 49  & 56  \\ \hline
   32   & 992  & 160   & 129  & 144   \\ \hline
64  & 4032  &  384 & 321 &  352\\ \hline
128 & 16256 & 896 & 769 & 832 \\ \hline
256 &65280 & 2048 &1793  &  1920\\ \hline
512 & 261632 & 4608 & 4097 &   4352\\ \hline
1024 & 1047552 & 10240 &  9217 & 9728 \\ \hline
2048 & 4192256 & 22528 &  20481 &  21504\\ \hline
4096 & 16773120 & 49152 &  45057 &  47104\\ \hline
\hline
\end{tabular}
\label{tbl:Vancost}
\end{table*}

\begin{table*}
\centering
\caption{Real GDB counts of the proposed radix-2 algorithms (i.e. $\mathbf{ vanc(z, N)}$ and $\mathbf{ vancc(z, N)}$) vs Direct computation}
\begin{tabular}{ | l | l | l | l | l |}
    \hline
   $N$ & Direct Add & $\#a\mathbb{R}(VanC, N)/$ & Direct Multi & $\#m\mathbb{R}(VanC, N)/$ \\
 &  & $\#a\mathbb{R}(VanCC, N)$ &  & $\#m\mathbb{R}(VanCC, N)$ \\
\hline \hline
4& 28  &  8 &   24&  8  \\ \hline
8& 120 &  24 &  112 &   30 \\ \hline
16  & 496 &  64 &  480&  90 \\ \hline
   32   &  2016 & 160   & 1984  &  242 \\ \hline
64  & 8128 &  384 & 8064 &  610 \\ \hline
128 & 32640 & 896 & 32512 & 1474\\ \hline
256 & 130816 & 2048 & 130560 & 3458 \\ \hline
512 &  523776& 4608 & 523264 &  7938 \\ \hline
1024 & 2096128  & 10240 &  2095104 &  17922 \\ \hline
2048 & 8386560 & 22528 &  8384512 &  39938 \\ \hline
4096 & 33550336 & 49152 & 33546240   &  88066 \\ \hline
\hline
\end{tabular}
\label{tbl:RVancost}
\end{table*}

\begin{table*}
\centering
\caption{Real GDB counts of the proposed radix-2 algorithms (i.e. $\mathbf{ vancr(z, N)}$ and $\mathbf{ vanccr(z, N)}$) vs Direct computation}
\begin{tabular}{ | l | l | l | l | l |}
    \hline
   $N$ & Direct Add & $\#a\mathbb{R}(VanCR, N)/$ & Direct Multi & $\#m\mathbb{R}(VanCR, N)/$ \\
 &  & $\#a\mathbb{R}(VanCCR, N)$ &  & $\#m\mathbb{R}(VanCCR, N)$ \\
\hline \hline
4& 28  &  8 &   24&  11  \\ \hline
8& 120 &  24 &  112 &   37 \\ \hline
16  & 496 &  64 &  480&  105 \\ \hline
   32   &  2016 & 160   & 1984  &  273 \\ \hline
64  & 8128 &  384 & 8064 &  673 \\ \hline
128 & 32640 & 896 & 32512 & 1601\\ \hline
256 & 130816 & 2048 & 130560 & 3713 \\ \hline
512 &  523776& 4608 & 523264 &  8449 \\ \hline
1024 & 2096128  & 10240 &  2095104 &  18945 \\ \hline
2048 & 8386560 & 22528 &  8384512 &  41985 \\ \hline
4096 & 33550336 & 49152 & 33546240   &  92161 \\ \hline
\hline
\end{tabular}
\label{tbl:RRVancost}
\end{table*}

Following Tables \ref{tbl:Vancost}, \ref{tbl:RVancost}, and \ref{tbl:RRVancost}, the proposed radix-2 algorithms for the Vandermonde matrices have shown significant arithmetic complexity reduction as opposed to the DVM algorithms presented in \cite{SVNA18, VS17, SAR19}. At the same time, we should recall that the DVM algorithms proposed in \cite{SVNA18, VS17, SAR19} have no restriction for nodes or delays as in this paper. Moreover, the proposed radix-2  algorithms for Vandermonde matrices have reduced GDB counts extensively opposed to the direct computation of Vandermonde matrices by a vector. More importantly, we have achieved the lowest GDB counts of radix-2 algorithms on computing Vandermonde matrices by a vector in the literature while covering radix-2 DFT algorithms as a subclass of the proposed radix-2 algorithms.

\section{Error Bound and Numerical Stability of Radix-2 Vandermonde Algorithms}
\label{sec:err}

\subsection{Theoretical Analysis}
Error bounds and numerical stability when computing the radix-2 Vandermonde algorithms associated with true time delays are the main concern in this section.
To derive analytic results for error bound,
we will use the perturbation of the product of matrices
(stated in~\cite{H96}).
Following the proposed radix-2 algorithms $\mathbf{ vancc(z, N)}$ and $\mathbf{ vanc(z, N)}$,
we have to compute weights $e^{\pm k(\frac{2 \pi j}{N})}=\omega_{\pm}^k(say)$, where $\omega_{\pm}=e^{\pm \frac{2 \pi j}{N}}$ for $k=0, 1, \ldots, \frac{N}{2}-1$.
The way we compute weights affects the accuracy of the algorithms.
Thus, we will assume that the computed weights
$\widehat{\omega}_{\pm}^k$ are used and satisfy for all
$k=0, 1, \ldots, \frac{N}{2}-1$
\begin{equation}
\widehat{\omega}_{\pm}^k = \omega_{\pm}^k + \epsilon_{k_{\pm}}, \:\:\: |\epsilon_{k_{+}}|\leq \mu_+, |\epsilon_{k_{-}}| \leq \mu_{-},
\label{scerror}
\end{equation}
where $\mu_{+}:=c_1u$ and$\mu_{-}:=c_1u$
$u$ is the unit roundoff,
and
$c_1$ and $c_2$ are constants that depend on the method~\cite{VL92}.

Let's recall the perturbation of the product of matrices
stated in~\cite[Lemma~3.7]{H96}
i.e. if $\mathbf{A}_k+\Delta \mathbf{A}_k \in \mathbb{R}^{N \times N}$ satisfies $| \Delta \mathbf{A}_k | \leq \delta_k |\mathbf{A}_k|$ for all $k$, then
\begin{equation}
\begin{aligned}
\begin{matrix}
& \Bigg| \displaystyle\prod_{k=0}^m \left( \mathbf{A}_k+\Delta \mathbf{A}_k\right)  - \displaystyle\prod_{k=0}^m \mathbf{A}_k   \Bigg|
 \leq
\\ & \hspace{1in}
\Bigg( \displaystyle\prod_{k=0}^m (1+\delta_k) -1 \Bigg) \displaystyle\prod_{k=0}^m \Bigg| \mathbf{A}_k \Bigg|
\end{matrix}
\nonumber
\end{aligned}
\end{equation}
where $|\delta_k| < u$.
Moreover, recall $\displaystyle\prod_{k=1}^N (1+\delta_k)^{\pm 1} = 1+\theta_N$ where $|\theta_N| \leq \frac{Nu}{1-Nu}=:\gamma_N$ and $\gamma_k+u \leq \gamma_{k+1}$,  $\gamma_k+\gamma_j+\gamma_k\gamma_j \leq \gamma_{k+j}$
from \cite[Lemma~3.1 and Lemma~3.3]{H96},
and for $x, y \in \mathbb{C}$, $fl(x \pm y)=(x +y)(1 + \delta)$ where $|\delta| \leq u$, $fl(xy)=(xy)(1 + \delta)$ where $|\delta| \leq \sqrt{2}\gamma_2$ from~\cite[Lemma~3.5]{H96}.

To carry out error analysis of the proposed algorithms in complex arithmetic, we implement complex arithmetic using real arithmetic operations computed according to number of additions and multiplications of non-unit numbers. Thus, we multiply $\hat{\mathbf{I}}_N$ (because it has only block identity matrices) and $\hat{\mathbf{D}}_N$, which were defined in \eqref{feqa}, and name as $\mathbf{B}_N$ s.t. $\mathbf{B}_N= \left[
\begin{array}{c|c}
 I_{\frac{N}{2}}&  I_{\frac{N}{2}}\\
\hline\\
 \dot{D}_{\frac{N}{2}} & -\dot{D}_{\frac{N}{2}}
\end{array}
\right] $. Similarly, we multiply $\hat{\mathbf{I}}_N$ (because it has only block identity matrices) and $\check{\mathbf{D}}_N$, which were defined in \eqref{feqa}, and name as $\check{\mathbf{B}}_N$ s.t. $\check{\mathbf{B}}_N=\left[
\begin{array}{c|c}
 I_{\frac{N}{2}}&  I_{\frac{N}{2}}\\
\hline\\
 \bar{\dot{D}}_{\frac{N}{2}} & -\bar{\dot{D}}_{\frac{N}{2}}
\end{array}
\right] $.

\begin{theorem}
\label{ThErR2V}
Let $\widehat{\mathbf{ y}}=fl(\mathbf{ V}_N \mathbf{ z})$, where $N=2^t(t \geq 2)$, be computed using the algorithm $\mathbf{ vancc(z, N)}$, and assume that \eqref{scerror} holds. Then
\begin{equation}
\begin{aligned}
\frac{\| \mathbf{ y} - \widehat{\mathbf{ y}} \|_2}{\|  \mathbf{ y} \|_2} \leq \frac{t \nu_{+}}{1- t \nu_+} N^{\frac{1}{2}}
\end{aligned}
\label{R2V}
\end{equation}
where $\nu_{+}=\eta_{+}\gamma_3 + \eta_{+}+ \gamma_3$ and $\eta_+=\mu_{+}+\gamma_4(1+\mu_{+})$.
\end{theorem}
\begin{proof}
Using the algorithm $\mathbf{ vancc(z, N)}$ and the computed matrices $\mathbf{\widehat{B}}(k)$ (in terms of computed weights $\widehat{\omega}_{+}^k$) for $k=0, 1, \cdots, t-2$: we have
\small
\begin{equation}
\begin{aligned}
\widehat{\mathbf{ y}}&=fl\Bigg(\mathbf{P}(0)\mathbf{P}(1) \cdots\mathbf{{P}}(t-2)\: \mathbf{V}(t-1)\:\widehat{\mathbf{B}}(t-2) \mathbf{C}(t-2)  \cdots
\\
& \hspace{.6in}
\widehat{\mathbf{B}}(1) \mathbf{C}(1) \widehat{\mathbf{B}}(0) \mathbf{{C}}(0)\: \mathbf{z}\Bigg)\\
&=\mathbf{P}(0)\mathbf{P}(1) \cdots\mathbf{P}(t-2)\: (\mathbf{V}(t-1) + \Delta\mathbf{V}(t-1))\:
\\
& \hspace{.2in}
(\widehat{\mathbf{B}}(t-2) + \Delta\widehat{\mathbf{B}}(t-2)) (\mathbf{ {C}}(t-2) + \Delta\mathbf{ C}(t-2))\cdots
\\
& \hspace{.3in}
(\widehat{\mathbf{B}}(1) + \Delta\widehat{\mathbf{B}}(1))(\mathbf{ {C}}(1)+ \Delta\mathbf{ {C}}(1))
\\
& \hspace{.4in}
(\widehat{\mathbf{B}}(0)+ \Delta\widehat{\mathbf{B}}(0)) (\mathbf{C}(0)+\Delta\mathbf{C}(0))\: \mathbf{ z}
\end{aligned}
\nonumber
\end{equation}
Each block diagonal matrix $\mathbf{ P}(k)$ and $\widehat{\mathbf{B}}(k)$ is formed by $2^k$ number of $P^T_{\frac{N}{2^k}}$'s and  $\mathbf{B}_{\frac{N}{2^k}}$'s respectively, in block diagonal positions. Using the fact that each $\mathbf{B}_{\frac{N}{2^k}}$ has only two non-zeros per row and recalling that we are using complex arithmetic, we get:
\begin{equation}
\begin{matrix}
\left | \Delta{\widehat{\mathbf{B}}(k)} \right | \leq {\gamma}_{4}\:\left | \widehat{\mathbf{B}}(k) \right |\:\:
{\rm for}\:\:\:k=0,1,\cdots,t-2.
\end{matrix}
\nonumber
\end{equation}
Using the fact that $\mathbf{\widehat{B}}(k)$ are computed using the computed weights $\widehat{\omega}_{+}^k$, we get:
\begin{equation}
\begin{matrix}
\widehat{\mathbf{B}}(k) =  \mathbf{ B}(k) + \Delta\mathbf{B}(k), \: \:\:\left | \Delta\mathbf{B}(k) \right | \leq \mu_{+} \left| \mathbf{ B}(k)\right|.
\end{matrix}
\nonumber
\end{equation}
Each block diagonal matrix $\mathbf{ C}(k)$ is formed by $2^k$ number of $\mathbf{C}_{\frac{N}{2^k}}$'s in block diagonal positions. Using the fact that each $\mathbf{C}_{\frac{N}{2^k}}$ has only one non-zeros per row and recalling that we are using complex arithmetic, we get:
\begin{equation}
\begin{matrix}
\left | \Delta{\mathbf{{C}}(k)} \right | \leq {\gamma}_{3}\:\left | \mathbf{{C}}(k) \right |\:\:
{\rm for}\:\:\:k=0,1,\cdots,t-2.
\end{matrix}
\nonumber
\end{equation}
$\mathbf{{V}}(t-1)$ is a block diagonal matrix and formed by $2^{t-1}$ number of $\mathbf{ V}_2$'s in diagonal positions. Hence
\begin{equation}
\begin{matrix}
\left | \Delta{\mathbf{{V}}(t-1)} \right | \leq {\gamma}_{3}\:\left | \mathbf{{V}}(t-1) \right |.
\end{matrix}
\nonumber
\end{equation}
Thus overall,
\small
\begin{equation}
\begin{aligned}
\widehat{\mathbf{ y}}&= \mathbf{{P}}(0)\mathbf{{P}}(1) \cdots\mathbf{{P}}(t-2)(\mathbf{{V}}(t-1) + \Delta\mathbf{{V}}(t-1))
\\
& \hspace{.2in}
(\mathbf{{B}}(t-2) + \mathbf{{E}}(t-2)) (\mathbf{ {C}}(t-2) + \Delta\mathbf{ {C}}(t-2))\cdots
\\
& \hspace{.3in}
(\mathbf{{B}}(1) + \mathbf{{E}}(1))(\mathbf{ {C}}(1)+ \Delta\mathbf{ {C}}(1))
\\
& \hspace{.4in}
(\mathbf{{B}}(0) + \mathbf{{E}}(0))(\mathbf{{C}}(0)+\Delta\mathbf{{C}}(0))\: \mathbf{ z}
\end{aligned}
\nonumber
\end{equation}
where $|\mathbf{{E}}(k)   | \leq (\mu_{+}+\gamma_4(1+\mu_{+})) | \mathbf{{B}}(k)  |=\eta_{+}| \mathbf{{B}}(k)  |$.
\\
Hence
\small
\begin{equation}
\begin{aligned}
| \mathbf{ y} - \widehat{\mathbf{ y}}|
& \leq [ (1+\eta_{+})^{t-1}(1+\gamma_3)^t -1 ] \mathbf{{P}}(0)\mathbf{{P}}(1) \cdots\mathbf{{P}}(t-2)
\\ & \hspace{.3in}
|\mathbf{{V}}(t-1)| |\mathbf{{B}}(t-2)| |\mathbf{ {C}}(t-2)|  \cdots |\mathbf{{B}}(1)| |\mathbf{ {C}}(1)|
\\ & \hspace{.4in}
|\mathbf{{B}}(0)| | \mathbf{{C}}(0)|  |\mathbf{ z}|.
\end{aligned}
\nonumber
\end{equation}
Since each $\mathbf{{C}}(k)$ is an unitary matrix, and each $\mathbf{{B}}(k)$ and $\mathbf{{V}}(t-1)$ are unitary matrices up to scaling, we get $\| \mathbf{{C}}(k) \|_2=1$ and $\|\mathbf{{B}}(k)\|_2= \|\mathbf{{V}}(t-1)\|_2=\sqrt{2}$. Hence,
\begin{equation}
\| \mathbf{ y} - \widehat{\mathbf{ y}} \|_2 \leq \frac{t \nu_{+}}{1- t \nu_+} 2^{t} \| \mathbf{z}  \|_2,
\nonumber
\end{equation}
where $\nu_{+}=\eta_{+}\gamma_3 + \eta_{+}+ \gamma_3$. Now following $\mathbf{ V}_N \mathbf{ V}_N^H = N \cdot I_N$, we get $\| \mathbf{ y} \|_2 = \sqrt{n} \| \mathbf{ z} \|_2$, and hence the result.
\end{proof}

\begin{corollary}
\label{CoThErR2V}
Let $\widehat{\mathbf{ y}}=fl(\mathbf{ V}_N \mathbf{ z})$, where $N=2^t(t \geq 2)$, be computed using the algorithm $\mathbf{ vancc(z, N)}$, and assume that \eqref{scerror} holds. Then the proposed radix-2 algorithm for Vandermonde matrices i.e.  $\mathbf{ vancc(z, N)}$ is numerically stable.
\end{corollary}
\begin{proof}
Theorem \ref{ThErR2V} immediately follows that the proposed radix-2 algorithm for Vandermonde matrices i.e.  $\mathbf{ vancc(z, N)}$ can be computed with tiny forward error provided that the weights i.e. ${\omega}_{+}^k$ are computed stably. On the other hand,  $\widehat{\mathbf{ y}} = \mathbf{ y}+ \Delta\mathbf{ y}=\mathbf{ V}_N \mathbf{ z} + \Delta\mathbf{ y}$. Thus, we get $\widehat{\mathbf{ y}} = \mathbf{ V}_N( \mathbf{ z}+ \Delta\mathbf{ z})$ and $\frac {\| \Delta\mathbf{ z} \|_2}{ \| \mathbf{z} \|_2}= \frac {\| \Delta\mathbf{ y} \|_2}{ \| \mathbf{y} \|_2}$. If we compute $\mathbf{ y}=\mathbf{ V}_N \mathbf{ z}$ using the brute force computation, we get
\[
|\mathbf{ y} - \widehat{\mathbf{ y}}  | \leq \gamma_{N+2} | \mathbf{ V}_N|  |\mathbf{ z}|.
\]
Since $\mathbf{ V}_N$ is unitary w. r. t. scaling, this immediately reduces to
\begin{equation}
\frac{\| \mathbf{ y} - \widehat{\mathbf{ y}} \|_2}{\|  \mathbf{ y} \|_2} \leq \gamma_{N+2} N^{\frac{1}{2}}.
\label{se}
\end{equation}
As $\mu_{+}=\mathcal{O}(u)$, the error (\ref{R2V}) of the proposed radix-2 algorithm is much more smaller than that in \eqref{se}. Thus, the proposed algorithm is backward stable. Hence, the proposed algorithm is numerically stable.
\end{proof}

\begin{theorem}
\label{Th2ErR2V}
Let $\widehat{\mathbf{ y}}=fl(\mathbf{ V}_N \mathbf{ z})$, where $N=2^t(t \geq 2)$, be computed using the algorithm $\mathbf{ vanc(z, N)}$, and assume that \eqref{scerror} holds. Then
\begin{equation}
\begin{aligned}
\frac{\| \mathbf{ y} - \widehat{\mathbf{ y}} \|_2}{\|  \mathbf{ y} \|_2} \leq \frac{t \nu_{-}}{1- t \nu_-} N^{\frac{1}{2}}
\end{aligned}
\label{R2V2}
\end{equation}
where $\nu_{-}=\eta_{-}\gamma_3 + \eta_{-}+ \gamma_3$ and $\eta_{-}=\mu_{-}+\gamma_4(1+\mu_{-})$.
\end{theorem}
\begin{proof}
The proof follows similar lines as that of Theorem \ref{ThErR2V} except $\widehat{\check{\mathbf{ B}}}(k)$, $\mathbf{ \bar{C}}(k)$, $\widehat{\omega}^k_{-}$, and $\mu_{-}$ instead of $\widehat{\mathbf{ B}}(k)$, $\mathbf{ {C}}(k)$, $\widehat{\omega}^k_{+}$, and $\mu_{+}$, respectively.
\end{proof}

\begin{corollary}
\label{Co2ThErR2V}
Let $\widehat{\mathbf{ y}}=fl(\mathbf{ V}_N \mathbf{ z})$, where $N=2^t(t \geq 2)$, be computed using the algorithm $\mathbf{ vanc(z, N)}$, and assume that \eqref{scerror} holds. Then the proposed radix-2 algorithm for Vandermonde matrices i.e.  $\mathbf{ vanc(z, N)}$ is numerically stable.
\end{corollary}
\begin{proof}
The proof follows similar lines as in Corollary \ref{CoThErR2V}.
\end{proof}

\subsection{Numerical Results}
We will now state numerical results in connection to the error bounds of the proposed radix-2 algorithms for Vandermonde matrices and compare the results with the error bound of the radix-2 FFT algorithm analyzed in \cite{H96}. %
With the help of the radix-2 factorization of the DFT matrices in \cite{VL92}, it was proved in \cite{H96} that the error bound on computing radix-2 FFT algorithm is given by;
\begin{equation}
\frac{\| \mathbf{ y} - \widehat{\mathbf{ y}} \|_2}{\|  \mathbf{ y} \|_2} \leq \frac{t \eta}{1- t \eta} N^{\frac{1}{2}}
\label{EF}
\end{equation}
where $\widehat{\mathbf{ y}}=fl(\mathbf{ F}_N \mathbf{ x})$, $\mathbf{ F}_N$ is the DFT matrix, $N=2^t$, $\eta=\mu + \gamma_4(1+\mu)$, and $\mu$ depends on the methods for computing the weights as specified in \cite{VL92}. We compare the error bounds of the proposed radix-2 algorithms for Vandermonde matrices shown in \eqref{R2V} and \eqref{R2V2} with the radix-2 FFT algorithm \eqref{EF} using MATLAB(R2014a version). In these calculations, we have chosen $\mu=\mu_{+}=\mu_{-}=10^{-15}$ and $\gamma_N=\frac{Nu}{1-Nu}$ where $N=2^t$ and $u$ is the machine precision. Since $\mu = \mathcal{O}(u)$, we have chosen $u=10^{-15}$. Table \ref{tbl:EB} shows the error bounds of the proposed radix-2 algorithms for Vandermonde matrices and radix-2 FFT algorithm in \cite{H96}.

\begin{table}[h]
\centering
\caption{Error bounds of the proposed radix-2 algorithms (i.e. $\mathbf{ vancc(z, N)}$ and $\mathbf{ vanc(z, N)}$) vs radix-2 FFT algorithm \cite{H96}}
\begin{tabular}{ | l | l | l | }
    \hline
    $N$ & Error Bound & Error Bound  \\
 &  $\mathbf{ vancc(z, N)}$/$\mathbf{ vanc(z, N)}$   & FFT  \\
\hline \hline
4&   $3.2 \times 10^{-14}$ &   $2 \times 10^{-14}$  \\ \hline
8&   $6.8 \times 10^{-14}$ &  $4.2 \times 10^{-14}$  \\ \hline
16  &   $ 1.3 \times 10^{-13}$ &  $8 \times 10^{-14}$ \\ \hline
   32    & $ 2.3 \times 10^{-13}$   & $1.4 \times 10^{-13}$   \\ \hline
64   &  $ 3.8 \times 10^{-13}$ & $2.4 \times 10^{-13}$  \\ \hline
128  & $ 6.3 \times 10^{-13}$ & $4 \times 10^{-13}$ \\ \hline
256  & $ 1 \times 10^{-12}$ & $6.4 \times 10^{-13}$ \\ \hline
512 & $ 1.6 \times 10^{-12}$ & $1 \times 10^{-12}$  \\ \hline
1024   & $ 2.6 \times 10^{-12}$ &  $1.6 \times 10^{-12}$  \\ \hline
2048  & $ 4 \times 10^{-12}$ &  $2.5 \times 10^{-12}$  \\ \hline
4096  & $ 6.1 \times 10^{-12}$ & $3.8 \times 10^{-12}$  \\ \hline
\hline
\end{tabular}
\label{tbl:EB}
\end{table}

Based on the numerical results shown in Table \ref{tbl:EB}, the proposed radix-2 algorithms for Vandermode matrices and radix-2 FFT algorithm have the same error orders except for $N=16,$ and $256$. Even with these two $N$ values, error orders of the proposed algorithms and FFT vary only by $10^{-1}$ and relatively very low. To sum up, Table \ref{tbl:EB} shows that the proposed radix-2 algorithms for Vandermonde matrices provide tiny forward errors.

\section{Signal Flow Graphs for Radix-2 Vandermonde Algorithms}
\label{sec:sfg}
In this section, we use signal flow graphs to illustrate the connection between algebraic operations used in sparse and orthogonal factorization of Vandermonde matrices with the fundamental signal flow graphs
(SFG) building blocks (i.e. adders and multipliers). We provide two signal flow graphs to show the simplicity of the proposed radix-2 algorithms for Vandermonde matrices. Being pivotal for efficient physical implementation in hardware, SFGs should represent a numerical algorithm in its fully factorized form in such a way that more sparse matrices are resulted and, as a consequence, less arithmetic operations demanded. Thus, Fig.~\ref{figure-8-pt-vanc} displays the SFG for the proposed $\mathbf{ vanc(z, N)}$ algorithm for the case $N=8$.
The recursive nature is evident as we express the 8-point SFG in terms of the 4- and 2-point SFGs. Notice that, the SFG of the $\mathbf{ vancc(z, N)}$ algorithm is not presented because the delays have been replaced with time advances that are not realizable in real-time circuits. But for the software implementation purposes, we have proposed $\mathbf{ vancc(z, N)}$ algorithm in Section \ref{sub:algovan} to effectively compute Vandermonde matrices.

\begin{figure}
\centering
\subfigure[8-point]{\includegraphics{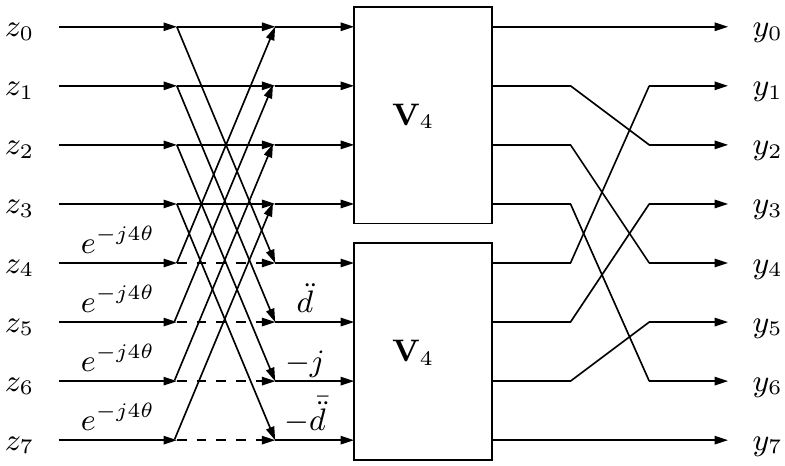}}
\subfigure[4-point]{\includegraphics{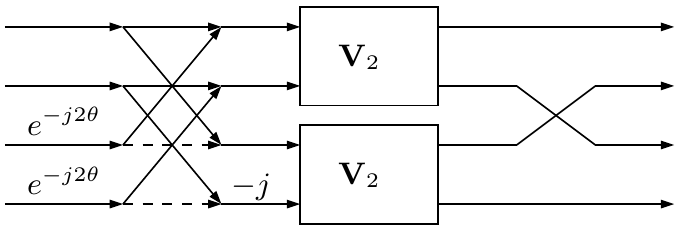}}
\subfigure[2-point]{\includegraphics{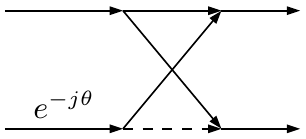}}
\caption{Signal flow graph of the 2-, 4-, and 8-point \textbf{vanc} decompositions, where $\ddot{d} = \frac{\sqrt{2}}{2} (1-j)$ and dashed arrows represent multiplication by $-1$.}
\label{figure-8-pt-vanc}
\end{figure}

\section{Conclusion}
\label{sec:con}
We have proposed novel self-recursive radix-2 algorithms for Vandermonde matrices. These algorithms have sparse and orthogonal factors. We have shown that the well known radix-2 DFT algorithm is a subclass of the proposed algorithms for the Vandermonde matrices.
The proposed algorithms attain the lowest gain-delay-block counts for Vandermonde matrices by a vector, in the literature.
Theoretical error bounds on computing the radix-2 algorithms and stability of the proposed algorithms are established. Numerical results of the forward error bounds of the proposed radix-2 algorithms are compared with the radix-2 FFT algorithm. The proposed radix-2 algorithms have shown tiny forward and backward errors when weights are computed stably.
Signal flow graphs were presented to show the simplicity of the proposed algorithm and to realize high-frequency analog circuits. Using the radix-2 algorithms for Vandermonde matrices associated with true time delay based delay-sum filterbanks, we have reduced the circuit complexity of multi-beam analog beamforming systems significantly.

\end{document}